\documentclass[a4paper]{article}

\usepackage[colorlinks,citecolor=blue,linkcolor=blue, urlcolor=blue]{hyperref}
\usepackage{amsmath, amsthm, amssymb}
\usepackage{color}
\usepackage{comment}
\usepackage[framemethod=tikz]{mdframed}
\usepackage{geometry}
\usepackage{enumitem}

\newtheorem{theorem}{Theorem}
\newtheorem{lemma}[theorem]{Lemma}

\newtheorem{corollary}[theorem]{Corollary}
\newtheorem{definition}[theorem]{Definition}

\newtheorem{remark}[theorem]{Remark}

\newcommand{\rab}{\rho_{AB}}
\newcommand{\ra}{\rho_A}
\newcommand{\ras}{{\rho_A^{\ast}}}
\newcommand{\sa}{\sigma_A}
\newcommand{\rb}{\rho_B}

\newcommand{\rtab}{\widetilde{\rho}_{AB}}
\newcommand{\rabal}{\rho_{AB}^{(\alpha)}}
\newcommand{\rtabal}{\widetilde{\rho}_{AB}^{(\alpha)}}
\newcommand{\sabp}{\sigma_{AB'}}

\newcommand{\sapbp}{\sigma_{A'B'}}
\newcommand{\rt}{\widetilde{\rho}}
\newcommand{\st}{\widetilde{\sigma}}

\newcommand{\rtabpq}{\widetilde{\rho}_{AB}^{(p,q)}}

\newcommand{\ket}[1]{|#1\rangle}
\newcommand{\bra}[1]{\langle #1 |}

\newcommand{\hilbert}[1]{\mathcal{H}_{#1}}

\newcommand{\lha}{\mathbf{L}(\mathcal{H}_A)}
\newcommand{\lhb}{\mathbf{L}(\mathcal{H}_B)}
\newcommand{\lhbp}{\mathbf{L}(\mathcal{H}_{B'})}
\newcommand{\reals}{\mathbb{R}}
\newcommand{\ribcb}[2]{\mathcal{R}_{{\text{\rm{CB,}}} #1 \rightarrow #2}}
\newcommand{\ribhc}[2]{{\mathcal{R}}_{#1 \rightarrow #2}}
\newcommand{\rib}[2]{{\mathcal{R}}_{#1 \rightarrow #2}}
\newcommand{\cbnorm}[3]{ \lVert #1 \rVert_{\text{\rm{CB}},#2 \rightarrow #3}}
\newcommand{\norm}[1]{ \Vert #1 \Vert}
\newcommand{\normt}[3]{\Vert #1 \Vert_{#2 \rightarrow #3}}

\newcommand{\tr}[1]{\text{\rm{tr}}\left(#1\right)}
\newcommand{\trA}[1]{\text{\rm{tr}}_A\left ( #1 \right)}

\newcommand{\trBp}[1]{\text{\rm{tr}}_{B'}\left ( #1 \right)}

\newcommand{\id}{\mathcal{I}}

\title{Impossibility of Local State Transformation via Hypercontractivity}

\author{ Payam Delgosha \\ {\it \small Department of Electrical Engineering,} {\it \small Sharif University of Technology,} {\it \small Tehran, Iran}\\{\it \small School of Mathematics,} {\it \small Institute for Research in Fundamental Sciences (IPM),} {\it \small Tehran, Iran}
\and Salman Beigi\\ {\it \small School of Mathematics,} {\it \small Institute for Research in Fundamental Sciences (IPM), Tehran, Iran}
\\}

\begin{document}
\maketitle

\begin{abstract} 
Local state transformation is the problem of transforming an arbitrary number of copies of a bipartite resource state to a bipartite target state under local operations. That is, given two bipartite states, is it possible to transform an arbitrary number of copies of one of them to one copy of the other state under local operations only?
This problem is a hard one in general since we assume that the number of copies of the resource state is arbitrarily large. In this paper we prove some bounds on this problem using the hypercontractivity properties of some super-operators corresponding to bipartite states. We measure hypercontractivity in terms of both the usual super-operator norms as well as completely bounded norms. 
\end{abstract}

\section{Introduction}
Local state transformation is the problem of transforming a given bipartite resource state $\rab$ to another bipartite target state $\sapbp$ under local operations only, i.e., do there exist completely-positive trace preserving (CPTP) super-operators $\Phi_{A\rightarrow A'}$ and $\Psi_{B\rightarrow B'}$ such that $\Phi\otimes \Psi(\rab)=\sapbp$? 
Solving this problem by brute-force search on the space of CPTP maps is not feasible when the dimensions of $\rab$ and $\sapbp$ are large, in which case imposing necessary conditions can be useful. 

Local operations do not generate entanglement, so if $\sapbp$ is more entangled than $\rab$ then such $\Phi, \Psi$ do not exist. So measures of entanglement provide us with bounds on the problem of local state transformation. Likewise, to attack this problem for classical states, we may use measures of correlation. For instance if the mutual information $I(A; B)$ of $\rab$ is less than that of $\sapbp$, then the latter cannot be 
generated from the former under local operations. 

These bounds however, usually fail when infinitely many copies of the resource state are available and we need to generate only \emph{one} copy of the target state, i.e., when we want to transform $\rab^{\otimes n}$, for a sufficiently large $n$, to $\sapbp$ under local operations. The point is that most measures of entanglement and correlation (such as mutual information, entanglement of formation, squashed entanglement etc.) tend to infinity on $\rab^{\otimes n}$ as $n$ gets larger and larger if $\rab$ is not uncorrelated or unentangled. Thus the following question arrises naturally: is there a measure of correlation or entanglement that give the same number to $\rab^{\otimes n}$ for all $n$?

\subsection{Maximal correlation}
There is a measure of correlation for bipartite classical states (distributions) called \emph{maximal correlation}. This measure is first introduced by Hirschfeld~\cite{Hirschfeld} and Gebelein~\cite{Gebelein} and then studied by R\'enyi~\cite{Renyi1, Renyi2}. Among other properties, maximal correlation satisfies data processing inequality. Namely, it does not increase under local operations. More importantly, maximal correlation of $n$ independent copies of a bipartite distribution is equal to the maximal correlation of only one copy. Given these two properties, maximal correlation gives a bound on the problem of local state transformation.

Maximal correlation has recently been defined for quantum states~\cite{Beigi12}. For a bipartite state $\rab$, its maximal correlation $\mu(\rab)$ is defined by 
\begin{align}
\mu(\rho_{AB})=\max\,\, & \left| \tr{\rho_{AB} X_A\otimes Y^{\dagger}_{B}}\right|\label{eq:max-corellation}\\
& \tr{\rho_AX_A} = \tr{\rho_B Y_B} =0,\nonumber\\
& \tr{\rho_A X_AX_A^{\dagger}} = \tr{\rho_B Y_BY_B^{\dagger}} =1.\nonumber
\end{align}
 This definition is reduced to the classical maximal correlation when $\rab$ is classical. It is shown in~\cite{Beigi12} that maximal correlation satisfies the following two important properties: 
 \begin{enumerate}
 \item[\rm{(i)}] $\mu(\rho_{A_1B_1}\otimes \rho'_{A_2B_2}) = \max\{\mu(\rho_{A_1B_1}), \mu(\rho'_{A_2B_2})\}$.
 \item[\rm{(ii)}] If $\sapbp = \Phi\otimes\Psi(\rab) $ then $\mu(\rab)\geq \mu(\sapbp)$.
 \end{enumerate}
As a result, using maximal correlation we may prove the impossibility of local state transformation in some cases even if infinitely many copies of the resource state is available. For example if we define 
\begin{align}\label{eq:rho-dep-1}
\rho^{(\alpha)}_{AB} = (1-\alpha) \frac{I_{AB}}{4} + \alpha\,\ket \psi\bra\psi_{AB},
\end{align}
where $\ket{\psi} =\frac{1}{\sqrt 2}(\ket{00}+\ket{11})$ and $I_{AB}/4$ is the maximally mixed state, then $\mu(\rab^{(\alpha)}) = \alpha$. This means that, if $\beta >\alpha$, having even infinitely many copies of $\rab^{(\alpha)}$ we cannot generate a single copy of $\rab^{(\beta)}$ under local operations.

Let us give another example. Let $\zeta_{UV}$ be the bipartite distribution over two bits defined by 
\begin{align}\label{eq:zeta-1}
\zeta_{00}=\zeta_{01}=\zeta_{10}=1/3 \quad \quad \text{ and }\quad \quad  \zeta_{11}=0.
\end{align}
 Then as pointed out in~\cite{KamathAnantharam} we have $\mu(\zeta_{UV}) = 1/2$. As a result, two parties who have shared  infinitely many copies of $\rab^{(\alpha)}$ cannot generate the bipartite correlation $\zeta_{UV}$ by local measurements if $\alpha<1/2$. 

Maximal correlation characterizes all states from which (perfect) shared  randomness can be distilled under local operations~\cite{Witsenhausen, Beigi12}. Nevertheless, as one expects, it does not solve the problem of local state transformation in general. In the above example we see that maximal correlation does not rule out the possibility of locally transforming $n$ copies of $\rab^{(\alpha)}$ to $\zeta_{UV}$ when $\alpha\geq 1/2$. The possibility of such a transformation for $\alpha=1$ is easily verified, but we do not know the answer for $1/2< \alpha<1$.

\subsection{Hypercontractivity}

Another idea to attack the problem of local state transformation is hypercontractivity. This idea is due to Ahlswede and G\'acs~\cite{AhlswedeGacs}, and has recently been revisited by Kamath and Anantharam~\cite{KamathAnantharam} and Anantharam et al.~\cite{Anantharametal} (see also~\cite{Raginsky} by Raginsky). Here is a rough description of the hypercontractivity method. 

Via Choi-Jamio{\l}kowski isomorphism a bipartite state $\rho_{AB}$ corresponds to a completely-positive super-operator $\Omega_{\rho}$ from the space of register $A$ to that of $B$. Suppose that for a CPTP map $\Psi_{B\rightarrow B'}$ we have $\sigma_{AB'} = \mathcal I\otimes \Psi(\rab)$, where $\mathcal I$ is the identity super-operator. This equation in terms of the corresponding super-operators gives $\Omega_{\sigma} = \Psi\circ \Omega_{\rho}$.

Recall that for every $1\leq p\leq \infty$ we may define a $p$-norm $\|\cdot\|_p$ (also called the Schatten norm).  Then for every $1\leq p, q\leq \infty$ we may consider the super-operator norm 
$$\|  \Omega_\rho  \|_{p\rightarrow q} = \sup_{X\neq 0} \frac{\|\Omega_\rho(X)\|_q}{\|X\|_p}.$$
From the definition of this norm and $\Omega_{\sigma} = \Psi\circ \Omega_{\rho}$ we obtain 
$$\|\Omega_{\sigma}\|_{p\rightarrow q} \leq \|\Psi\|_{q\rightarrow q} \|\Omega_\rho\|_{p\rightarrow q},$$
which puts restrictions on $\sigma$ in terms of super-operator norms. 

This restriction however is not very strong since $\|\Psi\|_{q\rightarrow q}$ could be very large. To overcome this problem instead of $\rab$ we may consider the super-operator corresponding to some \emph{normalization} of $\rab$ which we denote by $\rtab$. If the normalization is done properly, $\sigma_{AB'} = \mathcal I\otimes \Psi(\rab)$ will give $\Omega_{\st}=\widetilde \Psi\circ \Omega_{\rt}$, where again $\widetilde \Psi$ is some normalization of $\Psi$. As a result,
$$\|\Omega_{\st}\|_{p\rightarrow q} \leq \|\widetilde \Psi\|_{q\rightarrow q} \|\Omega_{\rt} \|_{p\rightarrow q}.$$
To get rid of the dependency on $\Psi$ in the above equation, the normalization is defined in such a way that 
\begin{align}\label{eq:psi-norm-1}
\|\widetilde \Psi\|_{q\rightarrow q}\leq 1,
\end{align}
for all values of $q$ and all CPTP maps $\Psi$. Putting these two inequalities together, we arrive at 
$$\|\Omega_{\st}\|_{p\rightarrow q} \leq \|\Omega_{\rt} \|_{p\rightarrow q}.$$

Yet this is not the final step since in the problem of local state transformation we assume that infinitely many copies of the resource state is available. Indeed we should compare the maximum of $\|\Omega_{\rt^{\otimes n}}\|_{p\rightarrow q}$ over all $n$, to $\|\Omega_{\st}\|_{p\rightarrow q}$. We have $\Omega_{\rt^{\otimes n}} = \Omega_{\rt}^{\otimes n}$ and by the definition of super-operator norm 
\begin{align}\label{eq:omega-norm-1}
\|\Omega_{\rt}^{\otimes n}\|_{p\rightarrow q} \geq  \|\Omega_{\rt}\|_{p\rightarrow q}^n.
\end{align}
Thus $\|\Omega_{\rt^{\otimes n}}\|_{p\rightarrow q}$ tends to infinity as $n\rightarrow \infty$ if $ \|\Omega_{\rt}\|_{p\rightarrow q}> 1$, in which case comparing to $ \|\Omega_{\st}\|_{p\rightarrow q}$ gives no bound. 

This observation suggests to consider the set of all pairs $(p, q)$ such that 
$$\|\Omega_{\rt}^{\otimes n}\|_{p\rightarrow q}\leq 1,$$
 for all $n$. This set is called the \emph{hypercontractivity ribbon}~\cite{KamathAnantharam}. Putting everything together we conclude that if $\sapbp$ can be locally generated from copies of $\rab$, then the hypercontractivity ribbon of $\sapbp$ is a subset of that of $\rab$.

\subsection{Quantum hypercontractivity ribbon}
The main contribution of this paper is to generalize the idea of hypercontractivity in the classical case~\cite{AhlswedeGacs, KamathAnantharam, Anantharametal} to the quantum setting. This idea is presented here based on the notation of the quantum theory, so this generalization may seem straightforward. Nevertheless there are some difficulties.
The first one is proving the upper bound on the super-operator norm of $\widetilde \Psi$ for all CPTP maps $\Psi$, i.e., equation~\eqref{eq:psi-norm-1}. This inequality in the classical case is a simple consequence of H\"older's inequality, but is highly non-trivial in the quantum case. Here we prove~\eqref{eq:psi-norm-1} based on the theory of complex interpolation and the Riesz-Thorin theorem. Another import difference is that inequality~\eqref{eq:omega-norm-1} is indeed an equality in the classical case. This equality simplifies very much the analysis of hypercontractivity.
In the quantum case however the super-operator norm is known to not be multiplicative even on completely-positive (CP) maps. We thus suggest to in addition to the usual super-operator norm, consider the completely bounded norm which is multiplicative on CP maps.

By generalizing the idea of hypercontractivity to the quantum setting we prove the impossibility of transforming $n$ copies of $\rab^{(\alpha)}$ defined in~\eqref{eq:rho-dep-1} to  $\zeta_{UV}$ defined by~\eqref{eq:zeta-1} under local operations if $\alpha <(1-\log 2/\log 3)^{1/2}\simeq 0.6075$. This result does not seem to be reproducible by other methods in quantum information theory.

We also study the relation between the hypercontractivity ribbon with some other measures of correlation. In particular we show that maximal correlation $\mu(\rab)$ gives a bound on the hypercontractivity ribbon of $\rab$. This result in the classical case is due to Ahlswede and G\'acs~\cite{AhlswedeGacs}.

The rest of this paper is organized as follows. In the following section we review the required tools including H\"older's inequalities, completely-bounded norms, the Riesz-Thorin theorem and Choi-Jamio{\l}kowski isomorphism. Section~\ref{sec:hyper-ribbon} includes the main definitions and main results of this paper. In particular the hypercontractivity ribbon is defined in this section and a data processing type property is proved. Some properties of the hypercontractivity ribbon, in particular its relation to the maximal correlation is discussed in this section. In Section~\ref{sec:example} we compute the hypercontractivity ribbon for some examples, and explain how log-Sobolev inequalities can be used to compute the ribbon. Concluding remarks come in Section~\ref{sec:conclusion}.

\section{Preliminaries}\label{sec:pre}
In this paper we assume that the reader is familiar with basic notions of quantum information theory~\cite{NielsenChuang} such as Hilbert spaces, Dirac's notation, density matrices, CPTP maps etc. Here we just fix some notations. 

We denote the Hilbert space corresponding to quantum register $A$ by $\hilbert{A}$, which throughout this paper is assumed to be finite dimensional. $\lha$ is the space of linear
operators acting on $\hilbert A$.
The identity operate acting on $\hilbert A$ is denoted
by $I_A\in \lha$.

Throughout this paper we fix an orthonormal basis $\{\ket 0, \ket 1, \dots,
\ket{d_A-1}\}$ for the Hilbert space $\hilbert A$  (the computational basis), where $d_A=\dim \hilbert A$. The transpose of $X\in \lha$with respect to this basis is denoted by
$X^T$, and
$X^*=(X^{\dagger})^T$ where $X^\dagger$ is the adjoint of $X$.

For a hermitian operator $X\in \lha$ we let $X^{-1}$ to be the inverse $X$ restricted to the support of $X$, i.e., $X^{-1}X= XX^{-1}$ is the hermitian projection on the span of eigenvectors of $X$ corresponding to non-zero eigenvalues. Furthermore, by $X\geq Y$ (and $Y\leq X$) we mean that $X-Y$ is positive semi-definite.


\subsection{Schatten norms}\label{sec:Schatten}
For $p\geq 1$ the Schatten $p$-norm of $X\in \lha$ is defined by 
$$\|X\|_p = \tr{|X|^p}^{1/p},$$
where $|X|:= (X^{\dagger}X)^{1/2}$. For $p=\infty$ we let
$$\|X\|_{\infty} = \lim_{p\rightarrow \infty} \| X  \|_p,$$
which is equal to the usual operator norm of $X$: 
$$\|X\|_{\infty} = \sup\{ \|X\ket v\|: \, \ket v\in \hilbert A, \|\ket v\|=1   
\}.$$
$\|\cdot\|_p$ satisfies triangle inequality and is a norm on $\lha$. We clearly have $\|X^T\|_p =
\|X^*\|_p=\|X^{\dagger}\|_p = \|X\|_p$. Moreover $\|UXV\|_p = \|X\|_p$ for all unitary operators $U, V$.

H\"{o}lder's inequality states that for every $1\leq p\leq \infty$,
$$\|XY\|_1 \leq \|X\|_p\|Y\|_{p'},$$
where $1\leq p'\leq \infty$ is the H\"older conjugate of $p$, i.e., 
\begin{align}\label{eq:conjugate}
\frac 1 p +
\frac{1}{p'} = 1.
\end{align} 
Moreover, $\|\cdot\|_{p'}$ is the dual norm of $\|\cdot\|_p$,
i.e., for every $X$,
$$\|X\|_p  =\sup_{\|Y\|_{p'}=1}  
|\tr{XY}|.$$

A generalization of H\"older's inequality (see for example \cite{Bhatia-M})
states that for every $r, p, q>0$ with
$\frac{1}{r}=\frac{1}{p}+\frac{1}{q}$,
$$\|XY\|_r\leq \|X\|_p \|Y\|_q.$$
Then by a simple induction we have 
$$\|X_1\dots X_k\|_{r}\leq \|X_1\|_{p_1}\cdots \| X_k \|_{p_k},$$
for all $r, p_1, \dots, p_k>0$ with
 $\frac{1}{r} = \frac{1}{p_1} + \cdots +
\frac{1}{p_k}$.

For $1\leq p, q\leq \infty$ the super-operator norm of $\Phi:
\lha \rightarrow \lhb$ is defined by 
\begin{align*}
\| \Phi \|_{p\rightarrow q}:= \sup_{X\neq 0} \frac{\|\Phi(X)\|_q}{\|X\|_p}.
\end{align*}
When $\Phi$ is CP, the supremum could be taken over positive semi-definite $X$
\cite{JWatrous, KAudenaert, Szarek} (see also \cite{DevetakJKR}), i.e.\
\begin{equation}
\label{eq:norm-positive-input}
 \normt{\Phi}{p}{q} = \sup_{X>0} \frac{\norm{\Phi(X)}_q}{\norm{X}_p}.
\end{equation}

Observe that by definition 
$$\|\Phi(X)\|_q\leq \|\Phi\|_{p\rightarrow q} \|X\|_p,$$
for all $X\in \lha$. Moreover, 
$$\|\Phi\circ \Psi\|_{p\rightarrow q} \leq \|\Psi\|_{p\rightarrow r}\|\Phi\|_{r\rightarrow q}.$$

The super-operator norm is not multiplicative. That is although the inequality 
\begin{align}\label{eq:norm-multip}
\|\Phi \otimes \Psi\|_{p\rightarrow q} \geq \| \Phi \|_{p\rightarrow q} \|\Psi\|_{p\rightarrow q},
\end{align}
is easily verified, the reverse inequality does not hold in general. It is indeed well-known in quantum information theory that the super-operator norm is not multiplicative even in the case of $p=q=1$. To obtain a multiplicative super-operator norm we may consider the
\emph{completely bounded} norms.

\subsection{Completely bounded norms\label{sec:CB}}

In this section we review completely bounded norms. For details we refer the
reader to \cite{DevetakJKR} and references there including
\cite{Pisierbook, Jungethesis} .

As mentioned above the $1\rightarrow 1$ super-operator norm is not multiplicative. To make it multiplicative we usually define the completely bounded norm as 
$$\cbnorm{\Phi}{1}{1} =\sup_d \norm{\mathcal I_d\otimes \Phi}_{1\rightarrow 1}=\sup_d \sup_{X\neq 0} \frac{\norm{\mathcal I_d\otimes \Phi(X)}_1}{\norm X_1},$$
where by $\mathcal I_d$ we mean the identity super-operator acting on a $d$-dimensional Hilbert space. This norm is also called the diamond norm \cite{KitaevSV}.

The completely bounded norm $p\rightarrow q$ can be defined similarly when $p=q$. Nevertheless, when $p\neq q$ this definition does not make sense since the supremum may not even exist. The point is that the norm $\norm{\id_d}_{p\rightarrow q}$ of the identity operator is not equal to $1$. 

To overcome this problem the completely bounded norm can be defined by
\begin{align}\label{eq:CB-pq}
\cbnorm{ \Phi}{p}{q} : = \sup_d \|  \mathcal I_d\otimes \Phi 
\|_{(t,p)\rightarrow (t,q)} = \sup_d \sup_X \frac{\|\mathcal I_d\otimes
\Phi(X)\|_{(t, q)}}{\|X\|_{(t, p)}}.
\end{align}
Here the $d$-dimensional auxiliary space on which $\id_d$ acts, is equipped with the $t$-Schatten norm, while the input and output spaces of $\Phi$ are equipped with $p$ and $q$ norms respectively. In fact, for every $X_{AB}\in \lha\otimes \lhb$ its $(t, p)$-norm can be defined via the theory of non-commutative vector valued $L_p$ spaces~\cite{Pisierbook}. If $X_{AB} = Y_A\otimes Z_B$ then 
\begin{align}\label{eq:tp-norm-special}
\|Y_A\otimes Z_B \|_{(t, p)} = \norm{Y_A}_t \norm{Z_B}_p.
\end{align}
But the definition of the $(t, p)$-norm for a general $X_{AB}$ is complicated and is deferred to Appendix~\ref{app:CB} since we do not require the exact form here. We instead review some basic properties of completely bounded norms.

In~\eqref{eq:CB-pq} the choice of $1\leq t\leq \infty$ is arbitrary;  For all values of $t$ we get the same number. In fact the completely bounded norm had been first defined for $t=\infty$, but then Pisier~\cite{Pisierbook} showed that all values of $t$ result in the same operator norm.

From the definition of completely bounded norm~\eqref{eq:CB-pq} we have
$$\cbnorm{\Phi\circ \Psi}{p}{q} \leq \cbnorm{\Psi}{p}{r}\cbnorm{\Phi}{r}{q}.$$
Also by considering $d=1$ in the definition we find that
$$\normt{\Phi}{p}{q}\leq \cbnorm{\Phi} p q.$$

 In this paper we use the following important theorem proved in~\cite{DevetakJKR}.

\begin{theorem}\label{thm:DJKR}
\begin{itemize}
\item[\rm(a)] For CP super-operators $\Phi, \Psi$ and
$1\leq p, q\leq \infty$ we have
$$\cbnorm{\Phi\otimes \Psi}{p}{q} = \cbnorm{\Phi}p q \cbnorm{\Psi}p q.$$

\item[\rm(b)] If $\Phi$ is CP, the supremum in \eqref{eq:CB-pq}
for every $d$ is achieved at a positive semi-definite $X$. 

\item[\rm(c)] If $1\leq q\leq p\leq \infty$ and $\Phi$ is CP
then
$$\cbnorm{\Phi} p q = \|\Phi\|_{p\rightarrow q}.$$
As a result, by \rm{(a)} the super-operator norm $\|\cdot\|_{p\rightarrow q}$ is multiplicative on CP maps when $q\leq p$.
\end{itemize}
\end{theorem}

\subsection{Riesz-Thorin theorem}

Most of our results are based on the theory of complex interpolation. Here we do not need this theory in detail, so we just give a brief review. For more details we refer the reader to~\cite{BerghL, Lunardi}.

Let $\mathcal V_0$ and $\mathcal V_1$ be two (complex) Banach spaces, i.e., two normed vector spaces that are complete under their norms. Moreover suppose that $\mathcal V_0, \mathcal V_1$ both can be embedded into a larger vector space. In this case $(\mathcal V_0, \mathcal V_1)$ is called an \emph{interpolation couple}.
Then in the theory of complex interpolation for every $0\leq \theta\leq 1$, a Banach space is constructed that is somehow an intermediate space in between $\mathcal V_0$ and $\mathcal V_1$. These spaces are denoted by 
$$[\mathcal V_0, \mathcal V_1]_{\theta}.$$

The main example of interpolation spaces is Schatten classes. Let $\mathbf L_p(\mathcal H)$ be the space of linear operators on the Hilbert space $\mathcal H$ equipped with the $p$-norm. Then it is well-known that for $1\leq p_0\leq p_1\leq \infty$ we have
\begin{align}\label{eq:interp-p-q}
[\mathbf L_{p_0}(\mathcal H), \mathbf L_{p_1}(\mathcal H)]_{\theta} = \mathbf L_{p_\theta}(\mathcal H),
\end{align}
where $p_\theta$ is given by
\begin{align}\label{eq:p-theta}
 \frac{1}{p_\theta} = \frac{1-\theta}{p_0} + \frac{\theta}{p_1}.
\end{align}

The other important example is the interpolation of $(p, q)$-norms mentioned in the previous section. If we let $\mathbf L_{(p,q)}(\mathcal H_{AB})$ to be the space $\mathbf L(\mathcal H_{AB})=\lha\otimes \lhb$ equipped with the $(p, q)$-norm then we have
$$[\mathbf L_{(p_0, q_0)}(\mathcal H_{AB}), \mathbf L_{(p_1, q_1)}(\mathcal H_{AB})]_{\theta} = \mathbf L_{(p_\theta, q_\theta)}(\mathcal H_{AB}),$$
where again $p_\theta$ and similarly $q_\theta$ are defined by~\eqref{eq:p-theta}.

The main result that we require from the theory of interpolation is the following variant of the Riesz-Thorin theorem taken from~\cite{Lunardi}. Here for ease of notation we state the theory only in the case where the corresponding Banach spaces (as sets) are subsets of finite dimensional matrices. 

Let
$$S:=\{z\in \mathbb C:  0\leq \text{Re} z\leq 1\}.$$ 
A map $f: S\rightarrow \mathbf L(\mathcal H)$ is called {holomorphic} (bounded, continuous) if the corresponding maps to matrix entries are holomorphic (bounded, continuous). 

\begin{theorem}
\label{thm:reisz-thorin}
Let $\mathcal V_0, \mathcal V_1\subseteq \mathbf L(\mathcal H_A)$ and $\mathcal W_0, \mathcal W_1\subseteq \lhb$ be interpolations couples. Suppose that for every $z\in S$ we have a super-operator 
$T_z: \lha \rightarrow \lhb$ such that for every $X\in \lha$ the map $z\mapsto T_z(X)$ is holomorphic and bounded in the interior of $S$ and continuous on the boundary. Then we may consider $T_z$ as a map from $\mathcal V_k$ to $\mathcal W_k$ (k=0,1) and consider its super-operator norm
$$\|T_z\|_{\mathcal V_k\rightarrow \mathcal W_k} = \sup_{X} \frac{\|T_z(X)\|_{\mathcal W_k}}{\| X \|_{\mathcal V_k}}.$$
Define 
\[
 M_0 = \sup_{t\in \reals} \normt{T_{it}}{\mathcal V_0}{\mathcal W_0}, \qquad M_1 = \sup_{t\in
\reals} \normt{T_{1+it}}{\mathcal V_1}{\mathcal W_1}.
\]
Then for $0\leq \theta \leq1$ we have
\[
 \normt{T_\theta}{\mathcal V_\theta}{\mathcal W_\theta} \leq M_0^{1-\theta} M_1^\theta,
\]
where
$
\mathcal V_\theta=[\mathcal V_0, \mathcal V_1]_{\theta}$ and $
\mathcal W_\theta=[\mathcal W_0, \mathcal W_1]_{\theta}.
$
\end{theorem}

Using~\eqref{eq:interp-p-q} this theorem in particular gives the following.

\begin{theorem}
\label{thm:reisz-thorin-special}
 For every $z \in S$, assume that $T_z: \mathbf L(\mathcal H_A) \rightarrow \mathbf L(\mathcal H_B)$ is a
linear operator such that for a fixed $X$, $z \mapsto T_z(X)$ is holomorphic and bounded in the
interior of $S$ and continuous on the boundary. Then for $0\leq \theta \leq 1$ we have
\[
 \normt{T_\theta}{p_\theta}{q_\theta} \leq \left(  \sup_{t\in \reals} \normt{T_{it}}{p_0}{q_0}        \right)^{1-\theta} \left(   \sup_{t\in
\reals} \normt{T_{1+it}}{p_1}{q_1} \right)^\theta,
\]
where $p_\theta, q_{\theta}$ are defined by~\eqref{eq:p-theta}.
\end{theorem}

In Appendix~\ref{app:reisz-thorin-schatten} we give a proof of the Riesz-Thorin theorem in the above special case. The proof in this case however captures the main ideas behind the Riesz-Thorin theorem in the general case.


\subsection{Choi-Jamio{\l}kowski isomorphism}

We denote the unnormalized maximally entangled state by
$$\ket{\chi}_{AA'} = \sum_{i=0}^{d_A-1} \ket i_A\ket i_{A'},$$
where $\{\ket 0, \ket 1, \dots, \ket{d_A-1}\}$ is a fixed orthonormal basis for
$\hilbert A$, and $\hilbert{A'}$ is a copy of $\hilbert{A}$.
For a given $\eta_{AB}\in \lha\otimes \lhb$ we may consider the corresponding
super-operator $\Omega_{\eta}: \lha\rightarrow \lhb$ via the
Choi-Jamio{\l}kowski isomorphism:
\begin{align*}
\eta_{AB} = \mathcal I_A\otimes \Omega_{\eta} \left(  \ket\chi\bra\chi_{AA'}  
\right) = \sum_{i,j=0}^{d_A-1} \ket i\bra j\otimes \Omega_{\eta}(\ket i\bra j).
\end{align*}
Then for every $X\in \lha$ we have 
\begin{align}\label{eq:omega-eta}
\Omega_\eta(X) = \trA{\eta_{AB}(X^T\otimes I_B)},
\end{align}
where $\text{tr}_A$ denotes the partial trace with respect to subsystem $A$.
This in particular implies that 
\begin{align}\label{eq:omega-eta-tr}
\tr{Y\Omega_\eta(X)} = \tr{ \eta_{AB}(X^T\otimes Y)}.
\end{align}
Moreover $\Omega_{\eta}$ is CP if and only if $\eta$ is
positive semi-definite.  Also observe that $\Omega_{\eta\otimes \eta'} = \Omega_{\eta}\otimes \Omega_{\eta'}$.

Using \eqref{eq:omega-eta-tr} and H\"older's duality we have
\begin{align}\label{eq:omega-eta-norm}
\normt{\Omega_\eta}p q  = \sup_{\|X\|_p=\norm Y_{q'}=1} |\tr{\eta_{AB}(X\otimes
Y)}|.
\end{align}
Similarly the completely bounded norm is computed as
\begin{align}\label{eq:omega-eta-cbnorm}
\cbnorm{\Omega_{\eta}} p q = \sup_d \sup_{\norm X_{(t,p)}=\norm Y_{(t', q')} =1}
|\tr{(\ket{\chi}\bra{\chi}_{CC'}\otimes \eta_{AB})(X_{CA}\otimes Y_{C'B})}|,
\end{align}
where $\hilbert C = \hilbert{C'}$ 
is a Hilbert space with dimension $d$. Here we use the fact that the dual norm
of $(t,p)$ is $(t',p')$ (see Appendix~\ref{app:CB} for more details) as
well as the fact that $\ket{\chi}\bra{\chi}$ is the Choi-Jamio{\l}kowski
representation of identity super-operator.

We now have all the tools required to state our results.

\section{The hypercontractivity ribbon}\label{sec:hyper-ribbon}

For a given bipartite state $\rab\in \lha\otimes \lhb$ and $1\leq p,q\leq \infty$ let
\begin{equation}
\label{eq:rtab}
 \rtabpq = \left ( \ra^{-\frac 1 {2p}} \otimes \rb^{-\frac 1 {2q}} \right )
\rab \left ( \ra^{-\frac 1 {2p}} \otimes \rb^{-\frac 1 {2q}} \right ).
\end{equation}
When it is clear from the context, we will drop the superscript $(p,q)$ and
subscript $AB$ and simply write $\rt$. From the definition the `tilde' operator corresponding to $\rab^{\otimes n}$ is equal to $\rtab^{\otimes n}$.

Since $\rtab$ is in  $ \lha \otimes \lhb$ we may
consider
the corresponding super-operator $\Omega_{\rt}: \lha \rightarrow \lhb$ via the
Choi-Jamio{\l}kowski representation:
\begin{equation}
\label{eq:rtab-choi}
 \rtab = \sum_{i,j = 0}^{d_A - 1} \ket{i} \bra{j}_A \otimes
\Omega_{\rt}(\ket{i}\bra{j})_B.
\end{equation}
Observe that since $\rtab$ for every density matrix $\rab$ is positive semi-definite, the corresponding super-operator
$\Omega_{\rt}$ is CP. 

\begin{definition}
For every bipartite quantum state $\rab$ and integer $n\geq 1$ define
\begin{equation}\label{eq:ribbon-def-1}
 \ribhc{A}{B}^{(n)}(\rab) := \left \{ (p,q') \in \reals^2 : q'\geq p \geq 1,
\normt{\Omega_{\rtab^{(p,q)}}^{\otimes n}}{p}{q'} 
\leq 1 \right \},
\end{equation}
and let
\begin{equation}\label{eq:ribbon-def-infty}
 \ribhc{A}{B}(\rab) := \bigcap_{n\geq 1} \ribhc{A}{B}^{(n)}(\rab).
\end{equation}
Moreover define
\begin{equation}\label{eq:ribbon-def-2}
  \ribcb{A}{B}(\rab) := \left \{ (p,q') \in \reals^2 :  q'\geq p \geq 1,
\cbnorm{\Omega_{\rtabpq}}{p}{q'} \leq 1 \right \}.
\end{equation}
Here $q'$ is the H\"older conjugate of $q$ defined by~\eqref{eq:conjugate}.
Following~\cite{KamathAnantharam, Anantharametal} we call $\rib A B(\rab)$ and $\ribcb A B(\rab)$ the hypercontractivity ribbons (HR).
\end{definition}

From the definition and~\eqref{eq:norm-multip} it is clear that 
$$\ribhc{A}{B}^{(nk)}(\rab)\subseteq \ribhc{A}{B}^{(k)}, $$
and then
\begin{align}\label{eq:tensorize-rib}
\ribhc A B(\rab^{\otimes k}) = \ribhc A B(\rab).
\end{align}
Moreover, by Theorem~\ref{thm:DJKR} the completely bounded norm is multiplicative on CP maps. Then we also have 
\begin{align}\label{eq:tensorize-cbrib}
\ribcb A B(\rab^{\otimes k}) = \ribcb A B(\rab).
\end{align}

Using~\eqref{eq:omega-eta-norm}, for a pair $(p, q')$ we have $(p, q')\in \ribhc{A}{B}^{(1)}(\rab)$ if and only if 
\begin{align}\label{eq:231}
\tr{\rtab^{(p, q)} X\otimes Y}\leq 1,
\end{align}
for all $X, Y$ such that $\norm X_p=\norm Y_q=1$. From this equation it is clear
that if we similar to $\Omega_{\rt^{(p,q)}}$ define a super-operator
$\Lambda_{\rt}:\lhb\rightarrow \lha$, then
$\norm{\Omega_{\rt^{(p,q)}}}_{p\rightarrow q'}\leq 1$ if and only if
$\norm{\Lambda_{\rt^{(p,q)}}}_{q\rightarrow p'}\leq 1$. 
In fact we may define the hypercontractivity ribbon \emph{from $B$ to $A$} by 
\begin{align*}
\ribhc{B}{A}^{(n)}(\rab) = \left\{ (q,p') \in \reals^2 : p' \geq q \geq 1,
\normt{\Lambda_{\rtab^{(p,q)}}^{\otimes n}}{q}{p'} 
\leq 1 \right\}.
\end{align*}
We then have 
\begin{align}\label{eq:ribbon-duality-1}
 (p, q')\in \ribhc{A}{B}^{(n)}(\rab) \quad \text{ iff }  \quad (q, p')\in \ribhc{B}{A}^{(n)}(\rab).
\end{align}
The same argument goes through for the completely bounded norm using~\eqref{eq:omega-eta-cbnorm}, so we have
\begin{align}\label{eq:ribbon-duality-2}
(p, q')\in \ribcb A B(\rab) \quad \text{ iff }  \quad (q, p')\in \ribcb B A(\rab).
\end{align}

\begin{remark}\label{rem:below-the-line}
In definitions~\eqref{eq:ribbon-def-1} and~\eqref{eq:ribbon-def-2} we put the restriction $q'\geq p$ which seems unnecessary. In 
Theorem~\ref{thm:below-the-line} we will justify this assumption by showing that for $q'\leq p$ the super-operator norms
$\normt{\Omega_{\rt^{(p, q)}}^{\otimes n}}{p}{q'}$ and $\cbnorm{\Omega_{\rt^{(p, q)}}}{p}{q'}$ are always equal to $1$, and give no information about $\rab$.
\end{remark}

\begin{remark} By the definition of the completely bounded norm and Theorem~\ref{thm:DJKR} we have 
 \begin{equation*}
 \normt{\Omega_{\rtab}^{\otimes n}}{p}{q'} \leq \cbnorm{\Omega_{\rtab}^{\otimes n}}{p}{q'} = \cbnorm{\Omega_{\rtab}}{p}{q'}^n.
\end{equation*}
Therefore if $(p, q')\in \ribcb{A}{B}{(\rab)}$ then $(p, q')\in \ribhc{A}{B}{(\rab)}$. In fact we always have
\begin{align}\label{eq:hc-subset-cb}
 \ribcb{A}{B}{(\rab)}\subseteq \ribhc{A}{B}{(\rab)}.
\end{align}
\end{remark}

\begin{remark}
Observe that
$$\tr{\rt^{(p,q)} \rho_A^{\frac 1 p}\otimes \rho_B^{\frac 1 q}}=1.$$
This means that in~\eqref{eq:231} if we let $X=\rho_A^{1/p}$ and
$Y=\rho_B^{1/q}$ we get equality. As a result for all $p, q$ both 
$\normt{\Omega_{\rt}} p {q'}$ and $\cbnorm{\Omega_{\rt}} p {q'}$ are at least
$1$. Therefore $(p, q')\in \ribhc A B(\rab)$ and $(p,q')\in \ribcb A
B(\rab)$ indeed mean $\normt{\Omega_{\rt}} p {q'}=1$ and $\cbnorm{\Omega_{\rt}}
p {q'}=1$ respectively. 
\end{remark}

\begin{remark} By Theorem \ref{thm:CBnorm-classical} the two ribbons $\ribhc A B(\rab)$ and $\ribcb A B(\rab)$ coincide when $\rab$ is classical.
\end{remark}

\subsection{Hypercontractivity ribbons under CPTP maps}

In this section by studying the behavior of HRs under CPTP maps we show that they are indeed measures of correlation. But before that we need to derive an expression for $\Omega_{\rt}$.

By definition we have 
\begin{equation}
\label{eq:rab-choi}
 \rab = \sum_{i,j=0}^{d_A-1} \ket{i} \bra{j} \otimes \Omega_\rho(\ket{i} \bra{j})
\end{equation}
Therefore,
\begin{align*}
\rtab &= (\ra^{-\frac1 {2p}} \otimes \rb^{-\frac 1 {2q}}) \rab (\ra^{-\frac 1 {2p}} \otimes
\rb^{-\frac 1 {2q}})  \\
&= \sum_{i,j} \ra^{-\frac 1{2p}} \ket{i}\bra{j} \ra^{-\frac 1{2p}} \otimes \rb^{-\frac 1{2q}}
\Omega_\rho ( \ket{i}\bra{j} ) \rb^{-\frac 1{2q}} \\
&= \sum_{i,j,k,l} \ket{k}\bra{k} \ra^{-\frac 1{2p}} \ket{i}\bra{j} \ra^{-\frac 1{2p}} \ket{l}
\bra{l} \otimes \rb^{-\frac 1{2q}} \Omega_\rho (\ket{i} \bra{j} ) \rb^{-\frac 1{2q}} \\
&= \sum_{i,j,k,l} \ket{k}\bra{i} \ras^{-\frac 1{2p}} \ket{k}\bra{l} \ras^{-\frac 1{2p}}
\ket{j} \bra{l} \otimes \rb^{-\frac 1{2q}} \Omega_\rho (\ket{i} \bra{j} ) \rb^{-\frac 1 {2q}}
\\
&= \sum_{i,j,k,l} \ket{k}   \bra{l} \otimes \rb^{-\frac 1{2q}} \Omega_\rho
(\ket{i}\bra{i} \ras^{-\frac 1{2p}} \ket{k}  \bra{l} \ras^{-\frac 1{2p}} \ket{j} \bra{j} )
\rb^{-1/2q} \\
&= \sum_{k,l} \ket{k}   \bra{l} \otimes \rb^{-\frac1 {2q}} \Omega_\rho ( \ras^{-\frac 1{2p}}
\ket{k}  \bra{l} \ras^{-\frac 1{2p}}) \rb^{-\frac1 {2q}}.
\end{align*}
Comparing to \eqref{eq:rtab-choi} we conclude that
\begin{equation*}
\Omega_{\rt}(X) = \rb^{-\frac 1{2q}} \Omega_\rho ( \ras ^{-\frac 1{2p}} X \ras^{-\frac 1{2p}} )
\rb^{-\frac 1{2q}}.
\end{equation*}

For positive semi-definite $\tau\in \lha$ define the super-operator $\Gamma_\tau:\lha
\rightarrow \lha$ by 
$$\Gamma_{\tau} (X) = \tau^{\frac{1}{2}} X \tau^{\frac 1 2}.$$
Then observe that $\Gamma_{\tau}^{\alpha} (X)= \tau^{\alpha/{2}} X
\tau^{\alpha/2}$ for all $\alpha$.
Using this notation we have
\begin{align}\label{eq:choi-rt-choi-r}
\Omega_{\rt}= \Gamma_{\rho_B}^{-\frac{1}{q}}\circ \Omega_{\rho}\circ \Gamma_{\rho_A^*}^{-\frac 1 {p}}.
\end{align}

Before proving the main result of this section we need the following important lemma.

\begin{lemma}\label{lem:norm-psi-tilde} For a CPTP map $\Phi$, density matrix $\tau$ and $1\leq p,q\leq \infty$
define 
\begin{align}\label{eq:Psi-tilde}
\widetilde\Phi=\widetilde{\Phi}^{(p,q)} = \Gamma_{\Phi(\tau)}^{-\frac 1 q} \circ \Phi\circ   \Gamma_{\tau}^{\frac{1}{p}}.
\end{align}
Then if $  q\geq p$,
\begin{align*}
\normt{\widetilde{\Phi}}{p'}{q'}= \cbnorm{\widetilde{\Phi}}{p'}{q'}\leq 1.
\end{align*}
\end{lemma}

\begin{proof} Observe that $\widetilde \Phi$ is CP. Then by part (c) of Theorem~\ref{thm:DJKR}, $\normt{\widetilde{\Phi}}{p'}{q'}= \cbnorm{\widetilde{\Phi}}{p'}{q'}$ since $q\geq p$ implies $p'\geq q'$. Thus we need to prove $\normt{\widetilde{\Phi}}{p'}{q'}\leq 1$.

Let $1\leq p_0 < p_1 \leq \infty$ and $1\leq q_0 < q_1\leq \infty$, and define $p_z$ and $q_z$ by
$$\frac{1}{p_z} = \frac{1-z}{p_0} + \frac{z}{p_1}\quad\quad \text{ and  } \quad \quad \frac{1}{q_z} = \frac{1-z}{q_0} + \frac{z}{q_1}.$$
Observe that 
$$\frac{1}{p'_z} = \frac{1-z}{p'_0} + \frac{z}{p'_1}\quad\quad \text{ and  } \quad \quad \frac{1}{q'_z} = \frac{1-z}{q'_0} + \frac{z}{q'_1}.$$
where $1/p_z+ 1/p'_z =1$ and $1/q_z+1/q'_z=1$.
Now define
$$T_z = \Gamma_{\Phi(\tau)}^{-\frac 1 {q_z}} \circ \Phi \circ \Gamma_\tau^{\frac 1{p_z}}.$$

$T_z$ satisfies the assumptions of Theorem~\ref{thm:reisz-thorin}. As a result, for every $0< \theta <1$ we have
$$\normt{T_\theta}{p'_\theta}{q'_\theta} \leq \left( \sup_{t\in \reals}   \normt{T_{it}}{p'_0}{q'_0}     \right)^{(1-\theta)} \left( \sup_{t\in \reals}   \normt{T_{1+it}}{p'_1}{q'_1}     \right)^{\theta}.$$
Observe that for every $t\in \reals$ there are unitaries $U, V$ such that 
$$T_{it}(X) = U T_0(VXV^{\dagger})U^{\dagger}.$$
Here we use the fact that $\tau$ and $\Phi(\tau)$ are hermitian and then
$\tau^{it}$ and $\Phi(\tau)^{it}$ are unitary. As a result we have
$\normt{T_{it}}{p'_0}{q'_0} = \normt{T_0}{p'_0}{q'_0}$. We similarly have 
$\normt{T_{1+it}}{p'_1}{q'_1} = \normt{T_1}{p'_1}{q'_1}$. Then we arrive at
\begin{align}\label{eq:phi-inter-theta}
\normt{T_\theta}{p'_\theta}{q'_\theta} \leq  \normt{T_{0}}{p'_0}{q'_0}^{(1-\theta)}  \normt{T_{1}}{p'_1}{q'_1}^{\theta}.
\end{align}
This means that if $\normt{T_{0}}{p'_0}{q'_0}$ and   $\normt{T_{1}}{p'_1}{q'_1}$ are at most $1$, then $\normt{T_\theta}{p'_\theta}{q'_\theta}$ is at most $1$ too.

Based on this observation if we prove the lemma in the special cases of $(p, q)= (1, q)$ and $(p, q)= (q, q)$ for arbitrary $1\leq q\leq \infty$, then we have the result for all $1\leq p\leq q\leq \infty$. For these two cases we can again use~\eqref{eq:phi-inter-theta}. If we prove the result for the three cases $(p, q)\in \{(1, 1), (1, \infty), (\infty, \infty)\}$ then we obtain a proof for all $1\leq p\leq q\leq \infty$. 

The case $(p, q)=(\infty, \infty)$ is verified noting that $\Phi$ is
completely positive, so using
\eqref{eq:norm-positive-input} we can restrict the supremum over positive
input, and also that $\Phi$ is trace preserving. For the case $p=1$ and
$q\in \{1, \infty\}$,
note that $\widetilde \Phi$ is CP, so the maximum of $\|\widetilde
\Phi(X)\|_{q'}$ over all $X$ with $\norm{X}_\infty=1$ is obtained at $X=I$
(see~\cite{Bhatia-P}).

\end{proof}

A special case of this lemma has also been proved in~\cite{Beigi13}, and has found other applications in quantum information theory

\begin{theorem}
 \label{thm:CB-data-processing}
For a CPTP map $\Phi : \lhb \rightarrow \lhbp$ let $\sabp = \mathcal I_A \otimes \Phi (\rab)$. Then for all $n\geq 1$ we have
\[
\ribhc{A}{B}^{(n)}(\rab)\subseteq \ribhc{A}{B'}^{(n)}(\sabp),
\]
and
\[
  \ribcb{A}{B}(\rab)\subseteq \ribcb{A}{B'}(\sabp).
\]

\end{theorem}

\begin{proof}
We need to show that for a pair $(p,q)$,
if $\cbnorm{\Omega_{\rt}}{p}{q'} \leq 1$ then
$\cbnorm{\Omega_{\st}}{p}{q'} \leq 1$, and that if $\normt{\Omega_{\rt}^{\otimes n}}{p}{q'} \leq 1$ then
 $\normt{\Omega_{\st}^{\otimes n}}{p}{q'} \leq 1$.
 
By
$$\sabp = \id \otimes \Phi (\rab) = \sum \ket{i}\bra{j} \otimes \Phi \big(
\Omega_\rho ( \ket{i} \bra{j} ) \big), $$ 
we have $\Omega_\sigma = \Phi \circ
\Omega_\rho.$ 
Moreover, $\sigma_A = \trBp{\sabp} = \trBp{\mathcal I\otimes \Phi(\rab)} = \rho_A$, and $\sigma_{B'} = \Phi(\rho_B)$.
Then using \eqref{eq:choi-rt-choi-r} we compute
\begin{align*}
\Omega_{\st} &= \Gamma_{\sigma_{B'}}^{-\frac 1 q} \circ \Omega_{\sigma}\circ \Gamma_{\sigma_A^*}^{-\frac 1 p}\\
&=  \Gamma_{\sigma_{B'}}^{-\frac 1 q} \circ \Phi\circ \Omega_{\rho}\circ \Gamma_{\sigma_A^*}^{-\frac 1 p}\\
& =  \left(\Gamma_{\Phi(\rho_B)}^{-\frac 1 q} \circ \Phi\circ   \Gamma_{\rho_B}^{\frac{1}{q}} \right)\circ \left(   \Gamma_{\rho_B}^{-\frac 1 q}\circ   \Omega_{\rho}\circ \Gamma_{\rho_A^*}^{-\frac 1 p}  \right)\\
& = \widetilde \Phi\circ \Omega_{\rt},
\end{align*}
where we set 
\begin{equation}
\label{eq:Psi}
\widetilde \Phi = \Gamma_{\Phi(\rho_B)}^{-\frac 1 q} \circ \Phi\circ   \Gamma_{\rho_B}^{\frac{1}{q}}.
\end{equation}
We now have
\begin{align*}
\cbnorm{\Omega_{\st}}{p}{q'} &= \cbnorm{\widetilde \Phi \circ \Omega_{\rt} }{p}{q'} \\
& \leq \cbnorm{\widetilde \Phi}{q'}{q'} \cbnorm{\Omega_{\rt}}{p}{q'} \\
& \leq  \cbnorm{\Omega_{\rt}}{p}{q'},
\end{align*}
where the last inequality is implied by Lemma~\ref{lem:norm-psi-tilde}.  As a result, if $(p, q')\in \ribcb A B(\rab)$ then $(p, q')\in \ribhc A B(\sabp)$.

The proof of $\normt{\Omega_{\st}^{\otimes n}}{p}{q'} \leq \normt{\Omega_{\rt}^{\otimes n}}{p}{q'}$ is identical noting that $\sabp^{\otimes n} = \mathcal I^{\otimes n}\otimes \Phi^{\otimes n} (\rab^{\otimes n})$.

\end{proof}

We are now ready to prove the main result of this paper.

\begin{corollary}\label{corol:data-processing}
Suppose that there are CPTP maps $\Phi_{A^n\rightarrow A'}$ and $\Psi_{B^n\rightarrow B'}$ such that $\sapbp= \Phi\otimes \Psi(\rab^{\otimes n})$. Then we have 
\[
\ribcb{A}{B} (\rab)\subseteq \ribcb{A'}{B'} ( \sapbp ),
\]
and
\[
\ribhc{A}{B}^{(n)} (\rab)\subseteq \ribhc{A'}{B'}^{(n)} ( \sapbp ),
\]
which in particular gives $\ribhc{A}{B}(\rab)\subseteq \ribhc{A'}{B'}( \sapbp )$.
\end{corollary}

\begin{proof} Let $\tau_{A^nB'}= \mathcal I_A^{\otimes n}\otimes \Psi (\rab^{\otimes n})$. Then we just need to prove 
\[
\ribhc{A}{B}^{(n)} (\rab)\subseteq \ribhc {A^n} {B'}^{(n)}(\tau_{A^nB'})\subseteq \ribhc{A'}{B'}^{(n)} ( \sapbp ),
\]
and
\[
\ribcb{A}{B} (\rab)\subseteq \ribcb{A^n}{B'}(\tau_{A^nB'})\subseteq  \ribcb{A'}{B'} ( \sapbp ).
\]
The first inclusions are straightforward consequences of Theorem~\ref{thm:CB-data-processing}. The second inclusions are similarly proved by exchanging the roles of registers $A$ and $B$ and using~\eqref{eq:ribbon-duality-1} and \eqref{eq:ribbon-duality-2}.
\end{proof}

We say that a resource state $\rab$ can be \emph{asymptotically} transformed to
$\sapbp$ under local transformations, if for 
every $\epsilon>0$ there exists $n$ and local
operations $\Phi_{A^n\rightarrow A'}$ and $\Psi_{B^n\rightarrow B'}$ such that
$$\norm{\sapbp - \Phi\otimes \Psi(\rab^{\otimes n})}_1\leq \epsilon.$$ 
Note that if $\Phi, \Psi$ exist for some $n$, then such local operations exist for all $m>n$ (simply ignore the extra $m-n$ copies of
$\rab$). 

\begin{corollary}\label{corol:assymp-transform}
Suppose that $\rab$ can be asymptotically transformed to $\sapbp$ under local transformations. Then we have 
\[
\ribcb{A}{B} (\rab)\subseteq \ribcb{A'}{B'} ( \sapbp ),
\]
and
\[
\ribhc{A}{B} (\rab)\subseteq \ribhc{A'}{B'} ( \sapbp ),
\]
\end{corollary}

\begin{proof} Let $\tau_{\epsilon}$ be the bipartite state for which there are $n$ and $\Phi, \Psi$ such that $\Phi\otimes \Psi(\rab^{\otimes n})=\tau_{\epsilon}$, and 
$$\norm{\sigma - \tau_{\epsilon}}_1\leq \epsilon.$$
Then $\tau_{\epsilon}$ tends to $\sigma$ as $\epsilon\rightarrow 0$ in $1$-norm. This implies that for every $p, q$ and $m$, $\widetilde \tau_{\epsilon}^{\otimes m}$ tends to $\st^{\otimes m}$ in $1$-norm, and in fact in any other norm. Here we use the fact that in \emph{finite} dimensions all norms are equivalent. This in particular gives that for every $p, q$ and $m$
$$\lim_{\epsilon\rightarrow 0} \normt{\Omega^{\otimes m}_{\widetilde \tau_{\epsilon}}}{p}{q'}=\normt{\Omega^{\otimes m}_{\st}}{p}{q'},$$
Therefore, if $(p,q') \in \ribhc{A}{B}(\rab)$, using
Corollary~\ref{corol:data-processing}, $(p,q')\in \mathcal{R}_{A\rightarrow B} (\tau_\epsilon)$,
hence $\normt{\Omega_{\widetilde{\tau_\epsilon}}^{\otimes m}}{p}{q'} \leq 1$ for
 all $m$. Fixing $m$ and taking the limit $\epsilon\rightarrow 0$ and using the above
 equality, for all $m$ we have
$\normt{\Omega_{\widetilde{\sigma}}^{\otimes m}}{p}{q'} \leq 1$. Thus
$(p,q')\in \ribhc{A'}{B'}(\sigma_{A'B'})$ and $\ribhc{A}{B}(\rab) \subseteq
\ribhc{A'}{B'}(\sigma_{A'B'})$.
Similarly we have
$$\lim_{\epsilon\rightarrow 0} \cbnorm{\Omega_{\widetilde
\tau_{\epsilon}}}{p}{q'}=\cbnorm{\Omega_{\st}}{p}{q'},$$
and again using Corollay~\ref{corol:data-processing} and by sending $\epsilon$ to
zero, we obtain $\ribcb{A}{B}(\rab) \subseteq \ribcb{A'}{B'}(\sigma_{A'B'})$.

\end{proof}

\subsection{Some properties of hypercontractivity ribbons}
In this section we further investigate properties of HRs. This properties may be useful in computing the ribbons and also to compare the ribbons, as measures of correlation, to other such measures. 

First as announced in Remark~\ref{rem:below-the-line} we justify the assumption $q' \geq p$ in the definition of HRs.

\begin{theorem}
\label{thm:below-the-line}
 For all $1\leq q'\leq  p \leq \infty $ we have 
$$\normt{\Omega_{\rt}}{p}{q'}=\cbnorm{\Omega_{\rt}}p{q'}\leq 1.$$ 
\end{theorem}

\begin{proof} Let 
$$\Phi = \Omega_{\rho}\circ \Gamma_{\rho_A^*}^{-1}.$$
$\Phi$ is obviously CP. Moreover according to~\eqref{eq:omega-eta} we have
\begin{align*}
\tr{\Phi(X)} &= \tr{ \rab\left(  \Gamma_{\rho_A^*}^{-1}(X) ^T\otimes I_B     \right)     }\\
& = \tr{ \rab\left(  (\rho_A^{-1/2}X^T\rho_A^{-1/2})\otimes I_B      \right)}\\
& = \tr{\rho_A\left(   \rho_A^{-1/2}X^T\rho_A^{-1/2}  \right)}\\
& = \tr{X}.
\end{align*}
This means that $\Phi$ is also trace preserving and then CPTP. Repeating similar calculations shows that $\Phi(\rho_A^{*}) = \rho_B$.
Then the proof is finished using Lemma~\ref{lem:norm-psi-tilde} and noting that 
$$\widetilde{\Phi}^{(p', q)} = \Gamma_{\Phi{(\rho_A^*)}}^{-\frac{1}{q}} \circ \Phi\circ \Gamma_{\rho_A^*}^{\frac{1}{p'}} = 
\Gamma_{\Phi{(\rho_A^*)}}^{-\frac{1}{q}} \circ \Omega_{\rho}\circ \Gamma_{\rho_A^*}^{-\frac{1}{p}} = \Omega_{\rt}.$$
\end{proof}

\begin{theorem}
\label{thm:convexity}
 The regions 
 \[
 \left\{ (1/p, 1/{q'}) : (p, q')\in \ribcb A B(\rab)\right\},
 \]
and 
 \[
 \left\{ (1/p, 1/{q'}) : (p, q')\in \ribhc{A}{B}^{(n)}(\rab)\right\},
 \]
 for every $n$ are convex.
\end{theorem}

\begin{proof} We should show that if $(p_0, q'_0), (p_1, q'_1)$ are in
$\ribhc{A}{B}^{(n)}(\rab)$ (or $\ribcb A B(\rab)$) then $(p_\theta, q'_\theta)$
is also in $\ribhc{A}{B}^{(n)}(\rab)$ (or $\ribcb A B(\rab)$) where $0< \theta
<1$ and 
$$\frac{1}{p_\theta} = \frac{1-\theta}{p_0} + \frac{\theta}{p_1}, \quad \quad \text{ and } \quad \quad \frac{1}{q'_\theta} = \frac{1-\theta}{q'_0} + \frac{\theta}{q'_1}.$$
The proof is based on Theorem~\ref{thm:reisz-thorin}. In fact following similar steps as in the proof of Lemma~\ref{lem:norm-psi-tilde} we obtain
$$\normt{\Omega^{\otimes n}_{\rt^{(p_\theta, q_\theta)}}}{p_\theta}{q'_\theta} \leq \normt{\Omega^{\otimes n}_{\rt^{(p_0, q_0)}}}{p_0}{q'_0}^{1- \theta} \normt{\Omega^{\otimes n}_{\rt^{(p_1, q_1)}}}{p_1}{q'_1}^{\theta},$$
and 
$$\cbnorm{\Omega_{\rt^{(p_\theta, q_\theta)}}}{p_\theta}{q'_\theta} \leq \cbnorm{\Omega_{\rt^{(p_0, q_0)}}}{p_0}{q'_0}^{1- \theta} \cbnorm{\Omega_{\rt^{(p_1, q_1)}}}{p_1}{q'_1}^{\theta}.$$
We are done. 
\end{proof}

The following lemma is sometimes useful in estimating HRs.

\begin{lemma}
\label{lem:rab-x-y}
 Assume that $(p,q')\in \ribhc{A}{B}^{(1)}(\rab)$. Then for $M,N \geq 0$ we
have
 \begin{equation}
 \label{eq:rab-x-y}
   \tr{\rab M\otimes N} \leq \tr{\ra M^p}^{1/p} \tr{\rb N^q}^{1/q}.
 \end{equation}
As a conclusion, if $0\leq M\leq I$ and $0\leq N\leq I$ we have,
 \[
   \tr{\rab M\otimes N} \leq \tr{\ra M}^{1/p} \tr{\rb N}^{1/q}.
 \]
\end{lemma}

\begin{proof}
We compute
\begin{align*}
   \tr{\rab M \otimes N} & = \tr{ \rtab \left (\ra^{1/2p} M \ra^{1/2p} \right )
\otimes \left (\rb^{1/2q} N \rb^{1/2q} \right )} \\
 &= \tr {\Omega_{\rt} \left ( {\ra^\ast}^{1/2p} M^T {\ra^\ast}^{1/2p} \right )
\rb^{1/2q} N \rb^{1/2q} }  \\
 & \leq \norm{\Omega_{\rt} \left ( {\ra^\ast}^{1/2p} M^T {\ra^\ast}^{1/2p}
\right )}_{q'} \norm{\rb^{1/2q} N \rb^{1/2q}}_q  \\
& \leq \normt{\Omega_{\rt^{(p,q)}}}{p}{q'} \norm{{\ra^\ast}^{1/2p} M^T {\ra^\ast}^{1/2p}}_p  \norm{\rb^{1/2q} N
\rb^{1/2q}}_q \\
&= \norm{{\ra}^{1/2p} M {\ra}^{1/2p}}_p  \norm{\rb^{1/2q} N
\rb^{1/2q}}_q \\
&\leq \tr{\ra M^p}^{1/p} \tr{ \rb N^q}^{1/q}.
\end{align*}
Here in the third line we use H\"older's inequality, and in the last line we use the Lieb-Thirring trace inequality~\cite{LiebThirring}. 
\end{proof}

The next theorem draws a connection between HR and the maximal
correlation. This statement in the classical case was first proved in \cite{AhlswedeGacs}.

\begin{theorem}\label{thm:ribbon-mu}
 Assume that $(p,q')\in \ribhc{A}{B}^{(1)}(\rab)$ and that the optimal operators  $X, Y$ in the  definition of  maximal correlation $\mu=\mu(\rab)$ in~\eqref{eq:max-corellation} are hermitian. Then we have
 \[
 \frac{p-1}{q'-1} = \frac{pq}{p'q'} \geq \mu^2.
 \]
\end{theorem}

Before giving a proof note that by this theorem, we note that under the assumption of the theorem, both the regions $\ribhc A
B(\rab)$  and $\ribcb A B(\rab)$ are in between lines $x=y$ and $x-1=
\mu^2 (y-1)$ in the real plane. This is the reason they are called
ribbon.

\begin{proof}
 Assume that $X$ and $Y$ are the optimal hermitian matrices that achieve the maximum in the
definition of maximal correlation. Therefore,
\begin{align}
  \mu = &| \tr{\rab X \otimes Y^{\dagger}} | \label{eq:mu-M-N}\\
  &\tr{\ra X} = \tr{\rb Y} = 0 \nonumber\\
  &\tr{\ra XX^\dagger} = \tr{\rb YY^\dagger} = 1.\nonumber
\end{align}

For $\alpha, \beta,x \in \reals$ let
\[
 M_x = I + x \alpha X, \qquad \text{ and } \qquad N_x = I + x \beta Y.
\]
Observe that $M_x, N_x$ are positive semi-definite for small enough $|x|$ (for fixed $\alpha, \beta$). Then by Lemma~\ref{lem:rab-x-y} we have
\[
 \tr{\rab M_x^{1/p} \otimes N_x^{1/q}} \leq \tr{\ra
M_x} ^{1/p} \tr{\rb N_x}^{1/q}.
\]
Using $\tr{\ra X} = \tr{\rb Y} = 0$ this inequality is simplified to
\[
 f(x) := \tr{\rab M_x^{1/p} \otimes N_x^{1/q}} - 1 \leq 0,
\]
for small $|x|$. Note that $f(0) = 0$ and
\begin{align*}
f'(x)=\frac{d}{dx} f(x) = \frac{\alpha}{p} \tr{ \rab \left (  X M_x^{\frac{1}{p}-1}  \otimes N_x^{1/q} \right ) }  + \frac{\beta}{q} \tr{ \rab \left ( M_x^{1/p} \otimes  Y
N_x^{\frac{1}{q}-1}  \right )}.
\end{align*}
To compute this derivative we use the fact that the pairs $X, M_x$ and $Y, N_x$ commute. As a result, $f'(0)=0$. Then using the fact that $f(x)$ is not positive in a neighborhood of $0$, we  should have $f''(0) \leq 0$. The second derivative of $f(x)$ is computed as
\begin{align*}
f''(x) = &\frac{\alpha^2}{p}\left ( \frac{1}{p} - 1 \right ) \tr { \rab \left ( 
X^2 M_x^{\frac{1}{p}-2} \otimes N_x^{1/q} \right )}\\
& + \frac{2\alpha\beta}{pq} \tr{ \rab \left (  XM_x^{\frac{1}{p}-1}
\otimes Y N_x^{\frac{1}{q} - 1} \right ) } \\
& + \frac{\beta^2}{q} \left ( \frac{1}{q} - 1 \right ) \tr{ \rab\left (
M_x^{1/p} \otimes Y^2 N_x^{\frac{1}{q} - 2} \right )}.
\end{align*}
Then using \eqref{eq:mu-M-N} and the fact that $X, Y$ are hermitian we have
\begin{align*}
 f''(0) &=  \alpha^2\frac{1}{p}\left ( \frac{1}{p} - 1 \right )  + \frac{2}{pq}
\alpha \beta \mu + \frac{1}{q} \left ( \frac{1}{q} - 1 \right ) \beta ^2\\
&  = -\alpha^2 \frac{1}{pp'}  + \alpha \beta\frac{2}{pq}
 \mu - \beta ^2\frac{1}{qq'}\\
 &   \leq 0.
\end{align*}
This inequality should hold for all $\alpha, \beta \in \reals$. Therefore the determinant of the coefficient matrix 
\[
\begin{pmatrix}
      \frac{1}{pp'} & -\frac{\mu}{pq} \\
      -\frac{\mu}{pq} & \frac{1}{qq'}
     \end{pmatrix}
\]
should be non-negative. This gives the desired result.
\end{proof}

Finally the following theorem in the classical case was proved in~\cite{AhlswedeGacs} (see also~\cite{Anantharametal}).

\begin{theorem}
 Assume that $(p,q')\in \ribhc{A}{B}^{(1)}(\rab)$ for $p,q' \geq 1$.
Then for all density matrices $\ra\neq \sa\in \lha$ we have
\begin{align}\label{eq:kl-ribbon}
 \frac{D(\sigma_B \Arrowvert \rb)}{D(\sa \Arrowvert \ra)} \leq \frac{q}{p'},
\end{align}
where $\sigma_B := \Omega_\rho \circ \Gamma_{\ra^*}^{-1} ( \sa^* )$. Here $D(\cdot \| \cdot)$ denotes the KL divergence.
\end{theorem}

\begin{proof} First note that as mentioned in the proof of
Theorem~\ref{thm:below-the-line}, 
$\Omega_\rho \circ \Gamma_{\ra^*}^{-1}$ is CPTP and $\Omega_{\rho}\circ
\Gamma_{\ra^*}^{-1}(\rho_A^{*}) = \rho_B$. Thus $\sigma_B := \Omega_\rho \circ
\Gamma_{\ra^*}^{-1} ( \sa^\ast )$ is also a density matrix.

Let $(p_1, q'_1)=(p, q')$ and $(p_0, q'_0)=(1, 1)$, and define 
\[
  \frac{1}{p_\theta} = \frac{1-\theta}{p_0} + \frac{\theta}{p_1}, \qquad \text{ and } \qquad \frac{1}{q'_\theta} = \frac{1-\theta}{q'_0} + \frac{\theta}{q'_1}. 
\]
By Theorem~\ref{thm:below-the-line} we have $\normt{\Omega_{\rt_{(p_0,q_0)}}}{p_0}{q'_0} \leq 1$. Then by Theorem~\ref{thm:convexity} we obtain
$\normt{\Omega_{\rt_{(p_\theta,q_\theta)}}}{p_\theta}{q'_\theta} \leq 1,$
for all $0\leq \theta\leq 1$. This means that for all $X$ we have
\[
 \norm{\Gamma_{\rb}^{-1/q_{\theta}} \circ \Omega_\rho \circ
\Gamma_{\ras}^{-1/p_{\theta}} (X) }_{q'_{\theta}} \leq \norm{X}_{p_{\theta}},
\]
or equivalently 
\[
 \norm{\Gamma_{\rb}^{-1/q_{\theta}} \circ \Omega_\rho \circ \Gamma_{\ra^*}^{-1}
 (X) }_{q'_{\theta}} \leq \norm{ \Gamma_{\ras}^{-1/p'_{\theta}}(X)}_{p_{\theta}}.
\]
In particular for $X= \sigma_A^*$ we find that 
\[
 h(\theta) = \norm{\Gamma_{\rho_A^*}^{-1/p'_{\theta}}(\sa^*)}_{p_\theta} -  \norm{\Gamma_{\rb}^{-1/q_{\theta}}  
 (\sigma_B) }_{q'_{\theta}}.
\]
is non-negative, i.e., $h(\theta)\geq 0$ for all $0\leq \theta\leq 1$. 

Observe that, 
\begin{align*}
 h(0) &= \norm{\sa^*}_{1} - \norm{\sigma_B}_{1}=0.
\end{align*}
Therefore we should have $h'(0) \geq 0$. 

The derivative $h'(\theta)$ can be computed using formulas provided in~\cite{OZ99} (see also \cite{KastoryanoTemme}).
Using these formulas we find that 
\begin{align*}
0 & \leq h'(0) \\
& = \frac{1}{p'_1} D(\sa^*\| \ra^*) -
\frac{1}{q_1} D( \Gamma_{\rb}^{-1}\circ \Omega_{\rt} (\sa^*)   \|   \Gamma_{\rb}^{-1}\circ \Omega_{\rt} (\sa^*) )  \\
& = \frac{1}{p'} D(\sa\| \ra) -
\frac{1}{q} D( \sigma_B  \|   \rb ).  \\
\end{align*}
We are done.
\end{proof}

Note that by Theorem~\ref{thm:below-the-line}, $\|\Omega_{\rt^(p,q)}\|_{p\rightarrow q'}\leq 1$ for $p=q'$. Then the above theorem in particular gives the data processing inequality for KL divergence.

\section{Some examples}\label{sec:example}

In this section we compute HRs for some bipartite states $\rab$. 
We first start with the extreme cases where $\rab$ is a product state and $\rab$ is a pure entangled state. In the former case $\rab$ contains no correlation and only product states can be generated from copies of $\rab$. Then Corollary~\ref{corol:assymp-transform} 
suggests that such a state should have the largest HRs. On the other hand, pure entangled states are the most correlated states so they should have the smallest HRs.

\subsection{Product states}

Assume that $\rab = \ra \otimes \rb$, thus
\[
 \begin{split}
  \rtab &= \left ( \ra^{-1/2p} \otimes \rb^{-1/2q} \right ) ( \ra \otimes \rb )
\left ( \ra^{-1/2p} \otimes \rb^{-1/2q} \right ) \\
&= \ra^{1/p'} \otimes \rb^{1/q'},
 \end{split}
\]
and by~\eqref{eq:omega-eta} we have
\begin{equation}
 \label{eq:tensorized-Omega}
  \Omega_{\rt}(X) = \tr{{\ra^*}^{1/p'} X} \rb^{1/q'}.
\end{equation}
Therefore,
\[
 (\id_C \otimes \Omega_{\rt}) (Y_{CA}) = \trA{ (I \otimes {\ra^*}^{1/p'} ) Y_{CA}} \otimes \rb^{1/q'}.
\]

Now we compute
 \begin{align*}
  \norm{(\id_C \otimes
\Omega_{\rt})(Y)}_{(t,q')}  &= \norm{\trA{ (I \otimes {\ra^*}^{1/p'} ) Y_{CA}} \otimes
\rb^{1/q'}}_{(t, q')} \\
&= \norm{\trA{ (I \otimes {\ra^*}^{1/p'} ) Y_{CA}}}_{t}
\norm{\rb^{1/q'}}_{q'} \\
&= \norm{\trA{ (I \otimes {\ra^*}^{1/p'} ) Y_{CA}}}_{t}\\
 &= \sup_{\norm{Z}_{t'} = 1}
\left | \tr{Z_C\, \trA{ (I \otimes {\ra^*}^{1/p'} ) Y_{CA}} } \right | \\
&= \sup_{\norm{Z}_{t'} = 1} \left | \tr{ ( Z_C \otimes {\ra^*}^{1/p'}) Y_{CA}} \right | \\
&\stackrel{(a)}{\leq} \sup_{\norm{X}_{(t',p')} = 1} \left | \tr{X_{CA}Y_{CA}} \right | \\
&\stackrel{(b)}{=} \norm{Y}_{(t,p)},
 \end{align*}
where in $(a)$ we use 
$$\norm{Z_C \otimes {\ra^*}^{1/p'}}_{(t',p')} = \norm{Z_C}_{t'} \norm{{\ra^*}^{1/p'}}_{p'} = 1$$ 
and $(b)$ holds because $(t', p')$ is the dual norm of $(t, p)$ (see Appendix~\ref{app:CB}). 

We conclude that 
$$\normt{\Omega_{\rt}}{p}{q'}\leq \cbnorm{\Omega_{\rt}}{p}{q'} \leq
1,$$
for all $p,q$. Thus
\[
 \ribhc{A}{B}(\ra\otimes \rb)=\ribcb{A}{B} (\ra \otimes \rb) = \{ (p,q') :  q'  \geq p \geq 1 \}.
\]

\subsection{Pure states}

Assume that $\rab=\ket{\varphi}\bra{\varphi}_{AB}$ is pure. If $\ket{\varphi}_{AB}$ is separable then its HRs are computed in the previous section. So let us assume that $\ket{\varphi}_{AB}$ is entangled. We claim that 
\begin{align}\label{eq:pure-ribbon}
\ribhc A B(\ket{\varphi}\bra{\varphi}_{AB}) = \ribcb A
B(\ket{\varphi}\bra{\varphi}_{AB}) = \{(p, q'): \, p=q'\geq 1 \}.
\end{align}
Indeed for every $p<q'$, by considering the Schmidt decomposition of $\ket{\varphi}_{AB}$ one can find (a rank-one) $X$ such that
$\norm{\Omega_{\rt}(X)}_{q'}> \norm{X}_p$. This implies~\eqref{eq:pure-ribbon}. Here we provide an indirect argument for this fact.

Suppose that $p\leq q'$ and $\normt{\Omega_{\rt}}{p}{q'}\leq 1$. As shown in~\cite{Beigi12} the maximal correlation of entangled pure states is equal to $\mu(\ket{\varphi}_{AB})=1$. 
Then by Theorem~\ref{thm:ribbon-mu} we have 
$$p-1\geq q'-1.$$
Given that $p\leq q'$ we find that $p=q'$. This gives~\eqref{eq:pure-ribbon}.

\subsection{Hypercontractivity via log-Sobolev inequalities}

Computing HRs is a hard problem in general. To compute $\ribhc{A}{B}(\rab)$ we should compute the norm $\normt{\Omega_{\rt}^{\otimes n}}{p}{q'}$ for all integers $n$. Likewise, computing $\ribcb{A}{B}(\rab)$ involves a supremum over the dimension of an auxiliary system which makes the computation of $\cbnorm{\Omega_{\rt}}{p}{q'}$ intractable. Here by giving an important example we show that quantum log-Sobolev inequalities~\cite{OZ99, KastoryanoTemme} provide useful tools for computing HRs.

Consider the bipartite state
 \[
  \rabal = \alpha \ket{\psi}\bra{\psi} + (1-\alpha) \frac{I_{AB}}{4},
 \]
where $\dim \hilbert{A} = \dim \hilbert{B} = 2$ and $\ket{\psi} = \frac{1}{\sqrt{2}} ( \ket{00} +
\ket{11})$ is the maximally entangled state.  Our goal in this section is to compute HRs of $\rabal$.

The maximal correlation of $\rabal$ can easily be computed~\cite{Beigi12}:
$$\mu(\rabal) = \alpha.$$
Then by Theorem~\ref{thm:ribbon-mu} for all $(p,q')\in \ribhc{A}{B}(\rabal)$ we have 
\begin{align}\label{eq:p-q-alpha}
\frac{p-1}{q'-1}\geq \alpha^2.
\end{align}

Note that independent of the value of $\alpha$, $\ra = \rb = I/2$.
Therefore,
\[
 \rtabal = 2^{(1/p + 1/q)} \rabal,
\]
and we have
\begin{align*}
\Omega_{\rt}(X) &=  2^{(1/p + 1/q)} \Omega_{\rho}(X) \\
& =2^{(1/p + 1/q)} \left( \alpha\Omega_{\ket{\psi}}(X) + (1-\alpha) \Omega_{I/4} (X)  \right)\\
& =2^{(1/p + 1/q)} \left( \frac{\alpha}{2} X + \frac{(1-\alpha)}{2} \tr{X} \frac{I}{2}  \right)\\
&= 2^{1/p - 1/q'} \Delta_\alpha(X),
\end{align*}
where $\Delta_\alpha$ denotes the depolarizing channel 
$$\Delta_{\alpha}(X)=\alpha X + (1-\alpha) \tr{X} I/{2}.$$

Let $\mathcal L$ be the super-operator defined by $\mathcal L(X) = X -\tr{X}\frac{I}{2}$. Then we have
$$e^{-t\mathcal L} = \Delta_{e^{-t}}.$$
That is, depolarizing channels belong to a semigroup of super-operators, and then their hypercontractivity can be studied based on log-Sobolev inequalities. Using these ideas King~\cite{King2012} proved the converse of~\eqref{eq:p-q-alpha}:
\[
(p, q')\in \ribhc A B(\rabal)  \qquad \text{ if and only if} \qquad
\alpha^2 \leq \frac{p-1}{q'-1}.
\]
This fact can be considered as a quantum analogue of Bonami-Beckner inequality from which one can compute the HR of the classical analogue of $\rabal$, i.e., a mixture of perfectly correlated coins and completely random coins \cite{KamathAnantharam}.

We know can resolve the problem mentioned in the introduction. Let $\zeta_{UV}$ be the bipartite distribution defined by
$$\zeta_{00}=\zeta_{01}=\zeta_{10}=1/3 \quad \quad \text{ and }\quad \quad  \zeta_{11}=0.$$
The maximal correlation of $\zeta_{UV}$ is equal to $\mu(\zeta_{UV})=1/2$. This means that if $\mu(\rabal)=\alpha<1/2$, then $\zeta_{UV}$ cannot be generated from copies of $\rabal$ under local measurement. Using hypercontractivity ribbons we now argue that this task is not doable for $\alpha\leq 0.6075$.

Suppose that by local measurement on $\sigma_{A^nB^n}={\rabal}\otimes \cdots \otimes \rabal$ we may generate $\zeta_{UV}$. That is, there are POVM measurements $\{M_0, M_1 = I - M_0\}$ and $\{N_0, N_1 = I-N_0\}$ such that for $i,j\in \{0,1\}$ we have
$$\zeta_{ij} = \tr{\sigma_{A^nB^n} M_i\otimes N_j }.$$
This in particular implies that
$$\zeta_{(U=i)} = \tr{\sigma_{A^n}  M_i}\qquad \text{ and }\qquad  \zeta_{(V=j)} = \tr{\sigma_{B^n}N_j}.$$
Let $(p, q')\in \ribhc{A}{B}(\rabal)\subseteq \ribhc{A}{B}^{(1)}(\rab)$. Then by Lemma~\ref{lem:rab-x-y} we obtain
\[
\tr{ \sigma_{A^nB^n} M_1 \otimes N_0} \leq \tr{\sigma_{A^n} M_1}^{1/p} \tr{\sigma_{B^n} N_0}^{1/q},
\]
or equivalently
\begin{align}\label{eq:1323}
 \frac{1}{3} \leq \left ( \frac{2}{3} \right ) ^{1/p} \left ( \frac{1}{3}
\right )^{1-1/q'}.
\end{align}
Note that $(p,q')=(1+k\alpha^2, 1+k)$ is in $ \ribhc{A}{B}(\rabal)$ for all $k\geq 0$.  Putting in~\eqref{eq:1323} we find that
$$3^{\frac{k(1-\alpha^2)}{k+1}} \leq 2, \qquad \forall k\geq 0.$$
But this inequality does not hold if 
$$\alpha< \sqrt{1-\frac{\log 2}{\log 3}}\simeq 0.6075.$$
 Observe that although, for instance, the maximal correlation of $\rab^{(0.6)}$ is greater than the maximally correlation of $\zeta_{UV}$, local transformation of $n$ copies of $\rab^{(0.6)}$ to $\zeta_{UV}$ is impossible even in the asymptotic limit.

 \section{Conclusion}\label{sec:conclusion}
 In this paper we defined two hypercontractivity ribbons, one corresponding to the usual super-operator norm $\rib A B(\rab)$, and the other corresponding to the completely bounded norm $\ribcb A B(\rab)$. By proving a data processing type property we concluded that these ribbons are indeed measures of bipartite correlation. These two ribbons coincide in the classical case, but we do not know of any quantum state $\rab$ such that $\rib{A}{B}(\rab)\neq\ribcb A B(\rab)$. Note that the completely bounded norm and the usual super-operator norm are really different~\cite{DevetakJKR}. 
 
 We also studied some properties of the ribbons. In particular we showed that
maximal correlation gives a bound on HRs. Moreover we proved a relation between
KL divergence and hypercontractivity ribbons. Here we should mention that in the
classical case the maximum of the left hand side of~\eqref{eq:kl-ribbon} over
all states $\sigma_A$ is equal to the infimum of the right hand side
over all $(p,q')\in \rib A B(\rab)$ (see~\cite{AhlswedeGacs} and
also~\cite{Anantharametal}). But we do not know whether such an equality holds
in the quantum case or not.

 The idea of \emph{reverse} hypercontractivity~\cite{Mosseletal} is applied in~\cite{KamathAnantharam} to study the ribbons for values $p, q<1$. We leave such an extension to the quantum case for future works.
  
It is argued in~\cite{Raginsky} that the hypercontractivity ribbon in the classical case can equivalently be characterized in terms of R\'enyi divergence. Given the recently proposed quantum R\'enyi divergence~\cite{Wildeetal, Lennertetal} it is not hard to see that such an equivalency holds in the quantum case as well.
  
In this paper we employed non-commutative vector valued Schatten spaces, and used completely bounded norms because the usual super-operator norm is not multiplicative in the quantum case. Such spaces have already been shown to be useful in quantum information theory~\cite{DevetakJKR, JungePalazuelos}. We hope that this work be another motivation for employing such normed spaces in quantum information theory.

 \vspace{.25in}
\noindent\textbf{Acknowledgements.} 
The authors are thankful to Robert Koenig for his valuable comments on an early version of this paper. 
SB was in part supported by National Elites Foundation and by a grant
from IPM (No. 91810409).


\appendix

\section{Non-commutative vector valued $L_p$ spaces\label{app:CB}}

For $1\leq q\leq p\leq \infty$ there exists $1\leq r\leq \infty$ such that
$\frac{1}{q} = \frac 1 p+ \frac 1 r$. Then for $X_{AB}\in \lha\otimes \lhb$ define 
\begin{align}\label{eq:pq-norm}
\|X\|_{(p,q)} := \sup_{U, V\in \lha} \frac{\| (U\otimes I_B) X(V\otimes I_B) 
\|_q}{ \|U\|_{2r} \|V\|_{2r}},
\end{align}
and 
\begin{align}\label{eq:qp-norm}
\|X\|_{(q, p)} := \inf_{X=(U\otimes I_B)Y(V\otimes I_B)} \| U  \|_{2r} \|V\|_{2r}
 \|Y\|_{p},
\end{align}
where in~\eqref{eq:qp-norm} the infimum is taken over all $U, V\in \lha$ and
$Y\in \lha\otimes \lhb$ such that $X=(U\otimes I_B)Y(V\otimes I_B)$. We can
compute $\|X\|_{(p,p)}$ (when $p=q$) from both~\eqref{eq:pq-norm}
and~\eqref{eq:qp-norm}, but there is no ambiguity here since they coincide.

$\|\cdot\|_{(p,q)}$ defined by equations \eqref{eq:pq-norm}
and~\eqref{eq:qp-norm} is indeed a norm on the tensor product space $\lha\otimes
\lhb$ for every $1\leq p, q\leq \infty$. It is clear from the definitions that 
$$\|X^T\|_{(p,q)} = \|X^*\|_{(p,q)}=\|X^{\dagger}\|_{(p,q)} = \|X\|_{(p,q)}.$$

Here we summarize some of the main properties of these norms. For proofs and details see \cite{Pisierbook, Xu, Jungethesis, DevetakJKR}. 

\begin{enumerate}[label=(\alph*)]

\item If $X_{AB}= M_A\otimes N_B$ then $\|X\|_{(p,q)} = \|M\|_{p}
\|N\|_{q}$.

\item $\|X\|_{(p,p)} = \|X\|_p$.

\item If  $X\in \lha\otimes \lhb$ is block diagonal with diagonal blocks $M_i\in \lhb$, i.e.,
$X= \sum_{i=0}^{d_A-1} \ket i\bra i\otimes M_i,$
then 
$$\| X \|_{(p,q)} = \left(  \sum_{i=0}^{d_A-1}  \|M_i\|_q^p  \right)^{\frac 1 p}.$$

\item If $X$ is positive semi-definite then in optimizations~\eqref{eq:pq-norm} and~\eqref{eq:qp-norm} we may assume that $U=V$ and that they are positive semi-definite.

\item $\|\cdot\|_{(p', q')}$ is the dual norm of $\|\cdot\|_{(p,
q)}$, i.e., for every $X$ we have 
\begin{align*}
\|X\|_{(p, q)}  =
\sup_{\|Y\|_{(p', q')}=1} |\tr{YX}|. 
\end{align*}
Here $p', q'$ are the H\"older conjugates of $p,q$ respectively, i.e. $1/p + 1/p'=1$.

\item
$\|X\|_{(p,q)} = \inf_{X=(U\otimes I_B)Y(V\otimes I_B)}
\|U\|_{2p}\|V\|_{2p} \|Y\|_{(\infty, q)}.$
\item
$\norm X_{(\infty, q)} = \sup_{U, V}  \frac{\norm{(U\otimes I) X(V\otimes
I)}_{(p, q)}}{\norm{U}_{2p}\norm{V}_{2p}}. 
$
\end{enumerate}

Now we can define the completely bounded norms as follows. For a super-operator
$\Phi:\lha \rightarrow \lhb$ and $t= \infty$ define
\begin{align}\label{eq:CB-pq-tt}
\cbnorm{ \Phi}{p}{q} : = \sup_d \|  \mathcal I_d\otimes \Phi 
\|_{(t,p)\rightarrow (t,q)} = \sup_d \sup_X \frac{\|\mathcal I_d\otimes
\Phi(X)\|_{(t, q)}}{\|X\|_{(t, p)}}.
\end{align}
Here $\mathcal I_d$ is the identity super-operator corresponding to a
$d$-dimensional Hilbert space.

In the definition of the completely bounded norm the value $t=\infty$ can be replaced with any $1\leq t\leq
\infty$;  no matter what $t$ is chosen we get to the same number~\cite{Pisierbook}. This fact can be easily proved using the last two properties mentioned above (see also \cite{Xu}).

We now show that the completely bounded norm and the usual super-operator norm coincide for classical channels.

\begin{theorem}\label{thm:CBnorm-classical}
Let $T:\lha\rightarrow \lhb$ be a classical channel of the form
$$T(\rho) = \sum_{i=0}^{d_A-1}  \sum_{j=0}^{d_B-1} c_{ij} \langle i|\rho \ket i \ket j\bra j,$$
where $c_{ij}\geq 0$.
Then for all $1\leq p, q\leq \infty$ we have $\cbnorm{T}{p}{q} = \normt{T}{p}{q}$.
\end{theorem}

\begin{proof} 

It suffices to show that $\|\mathcal I_C\otimes T\|_{(p,p)\rightarrow (p, q)} \leq \|T\|_{p\rightarrow q}$ for any finite dimensional space $\hilbert C$. Equivalently we show that 
\begin{align}\label{eq:c-c-norm}
\|\mathcal I_C\otimes T(\rho_{CA})\|_{(p,q)}\leq \normt{T}{p}{q} \|\rho_{CA}\|_{p}.
\end{align}

Note that by the `pinching inequality'~\cite{Bhatia-M}  we have $\|\rho_{CA}'\|_{p}\leq \|\rho_{CA}\|_p$ where
$$\rho'_{CA} = \sum_{i=0}^{d_A-1} (I\otimes \bra i) \rho_{CA}(I\otimes \ket i)\otimes \ket i\bra i,$$
and that $\mathcal I_C\otimes T(\rho_{CA}) = \mathcal I_C\otimes T(\rho'_{CA})$. As a result to prove~\eqref{eq:c-c-norm} we assume that $\rho_{CA}$ has the form
$$\rho_{CA} =  \sum_{i=0}^{d_A-1} \rho_i\otimes | i\rangle \langle i|_A.$$
Moreover, by Theorem~\ref{thm:DJKR} we may assume that $\rho_{CA}$ is positive semi-definite.
Then observe that
$$\|  \rho_{CA} \|_p = \left(  \sum_i \|\rho_i\|_p^p  \right)^{1/p} = \left(  \sum_i \left\|U_i\rho_iU_i^\dagger\right\|_p^p  \right)^{1/p} = \left\|\sum_{i}
U_i\rho_iU_i^{\dagger}\otimes \vert i\rangle \langle i\vert\right\|_p,$$
where $U_i$'s are arbitrary unitary matrices. 

Assume that $q\geq p\geq 1$, and let $\frac 1 p = \frac 1 q + \frac  1 r$. The other case where $p\geq q$ is similar (or one may use Theorem~\ref{thm:DJKR} in this case since $T$ is completely-positive). We now compute  
\begin{align*}
\left\| I\otimes T(\rho_{CA}) \right\|_{(p, q)} & = \left\| \sum_{i, j} c_{ij} \rho_i\otimes |j\rangle \langle j|  \right\|_{(p,q)} \\
&= \inf_{M\geq 0, \|M\|_{2r}=1}   \left\| \sum_{i, j} c_{ij} M^{-1}\rho_iM^{-1}\otimes |j\rangle \langle j|  \right\|_{q} \\
& \leq   \inf_{M\geq 0, \|M\|_{2r}=1}   \sup_{U_i} \left\| \sum_{i, j} c_{ij} M^{-1}U_i\rho_iU_{i}^{\dagger}M^{-1}\otimes |j\rangle \langle j|  \right\|_{q} \\
& =  \inf_{M\geq 0, \|M\|_{2r}=1}   \sup_{U_i} \left( \sum_j \left\| \sum_i  c_{ij} M^{-1}U_i\rho_iU_{i}^{\dagger}M^{-1}  \right\|_q^q \right)^{1/q}.
\end{align*}
By Lidskii's theorem~\cite{Bhatia-M} for every $M$, the term 
$$\left\| \sum_i  c_{ij} M^{-1}U_i\rho_iU_{i}^{\dagger}M^{-1}  \right\|_q$$
 is maximized when $U_i$'s are chosen in such a way that $U_i\rho_i U_i^{\dagger}$'s commute with $M$, and commute with each other. In other words we may assume from the beginning that $\rho_i$'s mutually commute.

So let us assume that $\rho_{CA} = \sum_{i,k} a_{ik} | k\rangle \langle k\vert_C \otimes |i\rangle \langle i|_A$. Then we have
\begin{align*}
\| I\otimes T(\rho_{CA})\|_{(p, q)} & = \left( \sum_k \left\|   \sum_{i,j} a_{ik} c_{ij} \vert j\rangle \langle j\vert  \right\|_q^p \right)^{1/p} \\
& = \left( \sum_k \left\|   T\left( \sum_{i} a_{ik} \vert i\rangle \langle i\vert\right)  \right\|_q^p \right)^{1/p} \\
& \leq \|T \|_{p\rightarrow q} \left( \sum_k \left\|    \sum_{i} a_{ik} \vert i\rangle \langle i\vert  \right\|_p^p \right)^{1/p} \\
& =   \|T \|_{p\rightarrow q}  \|\rho_{CA}\|_{p}.
\end{align*}

\end{proof}

\section{Proof of Riesz-Thorin theorem for Schatten norms\label{app:reisz-thorin-schatten}}

To prove this theorem we use Hadamard's three-line theorem~\cite{ReedSimon}.

\begin{theorem}\label{thm:hadamard} Let $f: S\rightarrow \mathbb C$ be a bounded function that is holomorphic in the interior of $S$ and continuous on the boundary. For $k=0, 1$ let
$$M_k=\sup_{t\in \mathbb R} |f(k+it)|.$$
Then for every $0\leq \theta\leq 1$ we have $|f(\theta)|\leq M_0^{1-\theta}M_1^{\theta}$.
\end{theorem}

\begin{proof}[Proof of Theorem~\ref{thm:reisz-thorin-special}]
First note that by H\"older's duality
 \[
  \normt{T_\theta}{p_\theta}{q_\theta} = \sup_{\norm{X}_{p_\theta} = 1,
\norm{Y}_{q'_\theta} = 1} | \tr{Y T_\theta(X)}|.
 \]
 So we need to show that for every $X,Y$ with $\norm{X}_{p_\theta} =\norm{Y}_{q'_\theta} = 1$ we have
\begin{align}\label{eq:yttheta}
|\tr{YT_{\theta}(X)}| \leq \left(  \sup_{t\in \reals} \normt{T_{it}}{p_0}{q_0}        \right)^{1-\theta} \left(    \sup_{t\in
\reals} \normt{T_{1+it}}{p_1}{q_1} \right)^\theta.
\end{align}

For an arbitrary matrix $X$ and complex number $z$ we may define $X^z$ using the singular value decomposition of $X$. That is, assume that $X=UDV$ where $U, V$ are unitary and $D$ is diagonal with non-negative entries. Then define $X^z:=U D^{z}V$. Using this notation, to prove~\eqref{eq:yttheta} we equivalently need to show that for every $X,Y$ such that $\|X\|_1=\|Y\|_1=1$ we have
\begin{align}\label{eq:yttheta-2}
\left|\tr{Y^{\frac{1}{q'_\theta}}T_{\theta}(X^{\frac{1}{p_\theta}})}\right| \leq \left(  \sup_{t\in \reals} \normt{T_{it}}{p_0}{q_0}        \right)^{1-\theta} \left(    \sup_{t\in
\reals} \normt{T_{1+it}}{p_1}{q_1} \right)^\theta.
\end{align}

Fix $X,Y$ with $\|X\|_1=\|Y\|_1=1$, and for $z\in S$ define
\[
 f(z) = \tr{ Y^{\frac{1-z}{q'_0} + \frac{z}{q'_1}}\, T_z(
X^{\frac{1-z}{p_0} + \frac{z}{p_1}} ) }.
\]
Note that 
$$f(\theta)=\tr{Y^{\frac{1}{q'_\theta}}T_{\theta}(X^{\frac{1}{p_\theta}})}.$$
$f(z)$ satisfies the assumptions of Hadamard's three-line theorem. Therefore we have
\begin{align}\label{eq:upper-h}
\left|\tr{Y^{\frac{1}{q'_\theta}}T_{\theta}(X^{\frac{1}{p_\theta}})}\right|=|f(\theta)| \leq \left( \sup_{t\in \mathbb R} |f(it)|    \right)^{1-\theta}\left(   \sup_{t\in \mathbb R} |f(1+it)|    \right)^{\theta}.
\end{align}

Observe that for $t\in \mathbb R$
\begin{align*}
  |f(it)| &= \left |  \tr{ Y^{\frac{1-it}{q'_0} + \frac{it}{q'_1}}\, T_{it}(
X^{\frac{1-it}{p_0} + \frac{it}{p_1}} ) }\right | \\
&\stackrel{(a)}{\leq} \norm{Y^{\frac{1-it}{q'_0}+\frac{it}{q'_1}}}_{q'_0}
\norm{T_{it} (X^{\frac{1-it}{p_0} + \frac{it}{p_1}} )}_{q_0} \\
& \stackrel{(b)}{\leq} \norm{Y^{\frac{1-it}{q'_0}+\frac{it}{q'_1}}}_{q'_0}
\normt{T_{it}}{p_0}{q_0} \norm{X^{\frac{1-it}{p_0} + \frac{it}{p_1}}}_{p_0} \\
& \stackrel{(c)}{\leq} \norm{Y^{1/q'_0}}_{q'_0} \normt{T_{it}}{p_0}{q_0} \norm{ X^{1/p_0}}_{p_0}
\\
& = \normt{T_{it}}{p_0}{q_0}.
 \end{align*}
Here in $(a)$ we use H\"older's inequality, in $(b)$ we use the definition of super-operator norm, and in 
$(c)$ we use the fact that $p$-norms are invariant under multiplication by unitary matrices.

We similarly have 
$$|f(1+it)| \leq \normt{T_{1+it}}{p_1}{q_1}.$$
Putting these two bounds in~\eqref{eq:upper-h} gives the desired result~\eqref{eq:yttheta-2}.

\end{proof}



\end{document}